\documentclass[runningheads]{llncs}
\usepackage{graphicx}
\usepackage{xspace}
\usepackage{graphicx}
\usepackage{xspace}
\usepackage{subfig}
\usepackage{lipsum}
\usepackage{amsmath}
\usepackage{amsfonts}
\usepackage{color}
\usepackage[utf8]{inputenc}
\usepackage{hyperref}
\usepackage{soul}

\usepackage{xcolor}
\definecolor{webred}{rgb}{0.5,0,0}
\definecolor{webblue}{rgb}{0,0,0.8}

\begin{document}
\title{The Black-Box Simplex Architecture for Runtime Assurance of Autonomous CPS}
\titlerunning{The Black-Box Simplex Architecture for Runtime Assurance of CPS}
%
\author{Usama Mehmood \and
Sanaz Sheikhi \and
Stanley Bak\and
Scott A.~Smolka\and\\
Scott D.~Stoller}
\authorrunning{U. Mehmood et al.}
%
\institute{Department of Computer Science, Stony Brook University, Stony Brook NY, USA\\
\email{\{umehmood, ssheikhi, sbak, sas, stoller\}@cs.stonybrook.edu}}
\maketitle              

\newcommand{\W}{\mathcal{W}}
\newcommand{\X}{\mathcal{X}}
\newcommand{\U}{\mathcal{U}}
\newcommand{\dmstep}{\ensuremath{dm_\textsf{step}}\xspace}
\newcommand{\dmupdate}{\ensuremath{dm_\textsf{update}}\xspace}

\newcommand{\um}[1]{#1}

\begin{abstract}
\newcommand{\blue}[1]{\textcolor{blue}{#1}}

The Simplex Architecture is a runtime assurance framework where control authority may switch from an unverified and potentially unsafe \emph{advanced controller} to a backup \emph{baseline controller} in order to maintain the safety of an autonomous cyber-physical system.
%
In this work, we show that runtime checks can replace the requirement to statically verify safety of the baseline controller.
This is important as there are many powerful control techniques, such as model-predictive control and neural network controllers, that work well in practice but are difficult to statically verify.
Since the method does not use internal information about the advanced or baseline controller, we call the approach the \emph{Black-Box Simplex Architecture}.
We prove the architecture is safe and present two case studies where (i)~model-predictive control provides safe multi-robot coordination, and (ii)~neural networks provably prevent collisions in groups of F-16 aircraft, despite the controllers occasionally outputting unsafe commands.
\end{abstract}

\newcommand{\blue}[1]{\textcolor{blue}{#1}}
\newcommand{\red}[1]{\textcolor{red}{#1}}
\newcommand{\orange}[1]{\textcolor{orange}{#1}}
\newcommand{\violet}[1]{\textcolor{violet}{#1}}
\newcommand{\bsa}{\textrm{BSA}\xspace}

\section{Introduction}
\label{sec:intro}
Autonomous cyber-physical systems (CPS) have the potential to transform vital domains such as transportation, health-care, and energy management. 
As these systems perform complex functions, they often require complex designs.
Moreover, since autonomous CPS interact with the physical world, they are typically safety-critical. 
Formal analysis, however, can be difficult for complex systems.

In the development of such CPS, powerful control techniques such as model-predictive control and deep reinforcement learning are increasingly being used instead of traditional 
controller design techniques. 
%
%
%
Such trends exacerbate the safety verification problem.
%
%
Additionally, there is increasing interest in systems that can \emph{learn in the field}, changing their behaviors based on observations.
Classical verification strategies are poorly suited for such designs.

One approach for dynamically providing safety for systems with complex and unverified components is \emph{runtime assurance}~\cite{clark2013}, where the state of the plant is monitored at runtime to mitigate possible imminent violations of formal properties. 
%
A well-known runtime assurance technique is the Simplex Control Architecture~\cite{seto1998simplex,sha2001}, which has been applied to a wide range of systems~\cite{desai2019_rsl,dung2017,schierman2015}.
%
%
In the original Simplex Architecture, shown in Figure~\ref{fig:simplex_bbsimplex}(a), the \emph{baseline controller} (BC) and the \emph{decision module} (DM) are part of the trusted computing base. 
%
%
The DM monitors the state of the system and switches control from the \emph{advanced controller} (AC) to the BC if using the former could result in a safety violation in the near future. 
%
%
The original Simplex Architecture requires creating a provably safe 
BC, which can be difficult.
%
%
In this work, we eliminate this requirement through
a greater reliance on
runtime verification.
%
%

In the proposed \emph{Black-Box Simplex Architecture} (\bsa), shown in Figure~\ref{fig:simplex_bbsimplex}(b), the BC (now referred to as the \emph{Lookahead Baseline Controller} (LBC)), no longer needs to be statically verified, and can even be incorrect.
The tradeoff is that the DM performs more extensive runtime checking and stores backup command sequences from previous computation steps.
The DM performs simulation or reachability analysis based on a known system model. If the DM's computation time is too large, \bsa keeps the system safe by switching control to a stored command sequence generated at an earlier step by the LBC and checked for safety by the DM.
The specifics of the approach will be discussed in Section~\ref{sec:BB}.


\begin{figure}[t]
\centering
\subfloat[Traditional Simplex Architecture]{\includegraphics[width=.50\columnwidth]{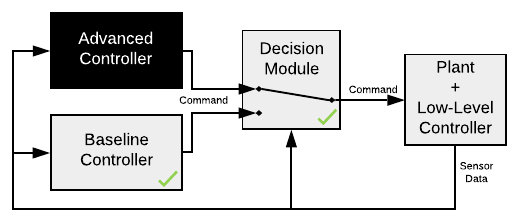}}\hfill
\subfloat[Black-Box Simplex Architecture]{\includegraphics[width=.50\columnwidth]{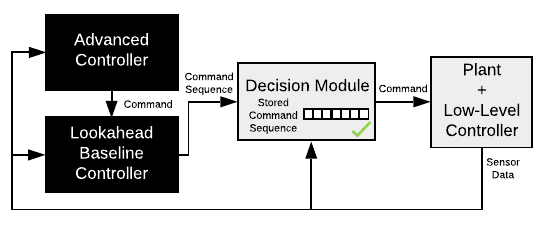}}\hfill
\caption{The Black-Box Simplex Architecture guarantees safety despite a black-box advanced controller and a black-box baseline controller.}\label{fig:simplex_bbsimplex} \vspace{-2em}
\end{figure}

%
%
%
%
%
%
%

We prove two theorems about this architecture: (i)~safety is always guaranteed, and (ii)~when the baseline and advanced controllers perform well (to be formally defined in Section~\ref{sec:BB}), the architecture is transparent: the advanced controller appears to have full control of the system.
The practicality of these assumptions and the utility of the \bsa architecture itself is
demonstrated through two significant
case studies.
In the first, a multi-robot coordination system uses a BC based on a model-predicative control algorithm with a potential-field approach for collision avoidance. 
Such a setup is difficult to statically verify as it depends on the online solution of a nonlinear optimization problem.
In the second, a mid-air collision avoidance system for groups of F-16 aircraft is created from imperfect logic encoded in neural networks.
A preview of the second case study is shown in Figure~\ref{fig:acasxu15}, where directly using the neural networks  causes a collision (left), but the Black-Box Simplex approach safely navigates the scenario, resulting in an emergent maneuver similar to a roundabout (right).

The rest of the paper is organized as follows. Section~\ref{sec:BB} presents a formal definition of the Black-Box Simplex Architecture, including proofs of safety and transparency. Section~\ref{sec:eval} features two case studies implementing the architecture. Section~\ref{sec:related} discusses related work and Section~\ref{sec:conclusions} offers our concluding remarks.

\begin{figure*}[t]
\centering
\subfloat[Original System (unsafe, the two red aircraft collide)]
{\includegraphics[width=.48\textwidth]{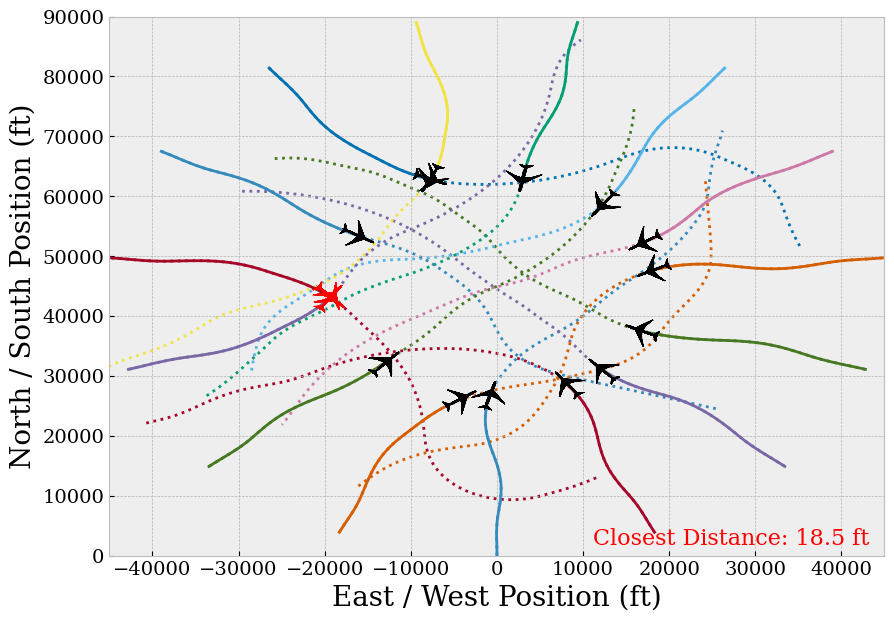}}\hfill
\subfloat[Black-Box Simplex (safe, snapshot shown at closest distance)]
{\includegraphics[width=.48\textwidth]{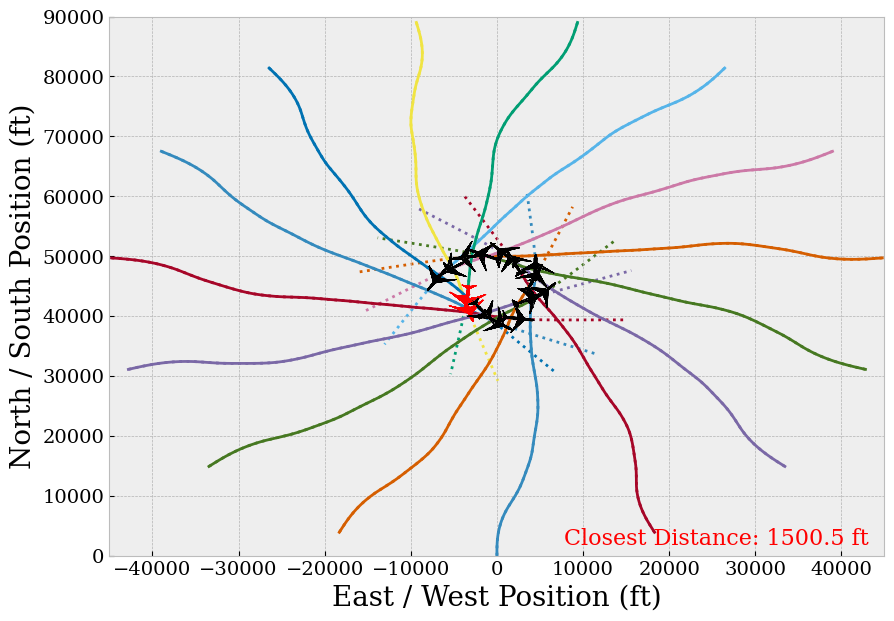}}
\caption{Black-Box Simplex safely navigates complex scenarios. In the 15-aircraft case, all aircraft cross the circle while maintaining a 1500 ft separation distance.}\label{fig:acasxu15} \vspace{-2em}
\end{figure*}
\section{Black-Box Simplex}
\label{sec:BB}

The traditional Simplex Architecture, shown in Figure~\ref{fig:simplex_bbsimplex}(a), preserves the safety of the system while permitting the use of an unverified AC. 
It does this by using the AC in conjunction with a verified BC and a verified DM. 
The DM cannot simply check if the next state is safe, as cyber-physical systems have inertia and it may be too late to take corrective action.
Rather, the verified design of a Simplex system usually requires offline reasoning with respect to a trusted BC and the system dynamics.

If the system dynamics are linear and the admissible states are defined with linear constraints,
a state-feedback BC and a DM can be synthesized by solving a linear matrix inequality~\cite{seto1998simplex}.
If the system dynamics or constraints are nonlinear, however, there is no direct approach to create a trusted BC and DM.
This prevents more widespread use of the traditional Simplex Architecture.

%
The proposed Black-Box Simplex Architecture removes the requirement that the BC is statically verified, allowing provable safety with both an unverified AC and an unverified BC.
Its architecture is shown in Figure~\ref{fig:simplex_bbsimplex}(b).
Apart from eliminating the need to establish safety of the BC, \bsa differs from the traditional Simplex Architecture in other important ways. 
First, the AC shares its command with the LBC instead of passing it directly to the DM.
Second, the LBC uses this command as the starting point of a \emph{candidate safe command sequence}. 
%

Candidate command sequences may be generated using state-of-the-art controller designs, including neural networks trained with reinforcement learning or MPC.  Note that a candidate command sequence is not guaranteed to be safe until it is verified by the DM through a runtime check.  
%
Specifically, the DM checks safety of the LBC's candidate command sequence, rejecting it if safety is not ensured. The DM checks safety by running simulations (rollouts) for deterministic systems; for systems with uncertainty, it performs online reachability computation~\cite{reachflow2020,bak2014rtss,althoff2014online}. 
BSA does not fail if the DM cannot finish the computation in time. 
Rather, it aborts the computation and switches to a backup command sequence that continues to ensure system safety.  
It can subsequently switch back to the AC when the runtime checks finish in time.

As long as the AC drives the system through states from which the LBC can recover, it continues to actuate the system.
However, if the LBC fails to compute a candidate command sequence that maintains safety---due to a fault of the unverified AC or the unverified BC, or due to excessive computation time for any of the components---the DM can still recover the system using the safe command sequence from the previous step.
Note that the DM does not generate any command sequences.  It only performs runtime checks and stores command sequences to maintain a safe backup plan at all times.

The applicability of BSA depends on the feasibility of two system-specific steps: (i)~constructing candidate command sequences and (ii)~proving their safety at runtime.
For some systems, a safe command sequence can simply bring the system to a stop.
%
An autonomous car, for example, could have a safe command sequence that steers the car to the side of the road and then stops.
A safe sequence for a drone might direct it to the closest emergency landing location.
For an rapidly-moving autonomous fixed-wing aircraft swarm, a safe sequence could fly all aircraft in non-intersecting circles to allow time for human intervention.
Proving safety of a given command sequence can also be challenging and depends on the system dynamics.
For nondeterministic systems, this could involve performing reachability computations at runtime~\cite{reachflow2020,bak2014rtss,althoff2014online}.
Such techniques assume an accurate system model is available in order to compute reachable sets.
Notice that traditional offline control theory also requires this assumption, so we do not view it as overly burdensome.

In \bsa, although both controllers are unverified, we do not combine them into a single unverified controller.
This allows for a logical separation of concerns, where the AC focuses on making progress on the mission, and the BC focuses on generating safe backup plans.



\subsection{Formal Definition of Black-Box Simplex}
\label{sec:bbdefs}

We formalize the behavior and requirements for the components of the Black-Box Simplex Architecture in order to prove properties about the system's behavior.

\textbf{Plant Model.} We consider discrete-time plant dynamics, modeled as a function
\begin{equation}
\label{eq:control_system}
    f(\underbrace{x_i}_{\text{state}}, \underbrace{u_i}_{\text{input}}, \underbrace{w_i}_{\text{disturbance}}) = \underbrace{x_{i+1}}_{\text{next state}}
\end{equation}
where $i \in \mathbb{Z}^{+}$ is the time step, $x_i \in \X$ is the system state, $u_i \in \U$ is a control input command, and $w_i \in \W$ is an environmental disturbance.
We sometimes also consider a \emph{deterministic} version of the system, where the disturbance $w_i$ can be taken to be zero at every step.

\textbf{Admissible States.} The system is characterized by a set of operational constraints which include physical limits and safety properties. 
States that satisfy all the operational constraints are called \emph{admissible states}.

\textbf{Candidate Command Sequences.} 
A single-input command is some $u \in \U$, and a $k$-length sequence of commands is written as $\overline{u} \in \mathcal{U}^k$. The length of a sequence can be written as $\overline{u}_{\text{len}} = k$, where we also can take the length of a single command, $u_{\text{len}} = 1$. We use Python-like notation for subsequences, where the first element in a sequence is $\overline{u}[0]$, and the rest of the sequence is $\overline{u}[\text{1:}]$.




\textbf{Decision Module.} The decision module in Black-Box Simplex stores a command sequence $\overline{s}$, which we sometimes call the decision module's state. 
The behavior of the DM is defined through two functions, \dmupdate and \dmstep.
The \dmupdate function
attempts to modify the DM's stored command sequence: 
\begin{equation}
\label{eq:dm_update}
    \dmupdate(\underbrace{x}_{\text{state}}, \underbrace{\overline{s}}_{\text{cur seq}}, \underbrace{\overline{t}}_{\text{proposed seq}}) = \underbrace{\overline{s'}}_{\text{new seq}}
\end{equation}
where if $\overline{s'} = \overline{t}$ then we say that the proposed command sequence is \emph{accepted}; otherwise $\overline{s'} = \overline{s}$ and we say that it is \emph{rejected}.  Correctness conditions on \dmupdate are given in Section~\ref{sec:theorems}.
\um{Note that the DM will accept a safe command sequence from the AC even if the previous command sequence from the AC was rejected because it was unsafe. As in~\cite{dung2019_nsa}, we refer to this as {\em reverse switching}, since it switches control back to the AC.}

The \dmstep function produces the next command $u$ to apply to the plant, as well as the next step's command sequence $\overline{s'}$ for the DM:
\begin{equation}
\label{eq:dm_step}
    \dmstep(\underbrace{\bar{s}}_{\text{cur seq}}) = (\underbrace{u}_{\text{next input}}, \underbrace{\overline{s'}}_{\text{next seq}})
\end{equation}
where $u = \overline{s}[0]$ and $\overline{s'}$ is constructed from $\bar{s}$ by removing the first command \um{(if the current sequence $\overline{s}$ has only one command then it is repeated)}:
\begin{equation*}
    \overline{s'} = \begin{cases}   \overline{s} &\text{if } \overline{s}_{\text{len}} = 1 \\  \overline{s}[\text{1:}] &\text{otherwise} \end{cases}
\end{equation*}

\textbf{Controllers.} The AC and LBC are defined using functions of the system state. 
\um{In particular, the AC is defined by a function $ac(x) = u$, where $u \in \mathcal{U}$ is a single command.}
\bsa's \emph{look-ahead baseline controller} is defined by $lbc(x) = \overline{u}$, \um{where $\overline{u} \in \mathcal{U}^k$ is a $k$-length command sequence}. The LBC outputs candidate command sequences that start with a given command, specifically, the command proposed by the AC. These can be defined with a function $\textit{lbc}_{ac}(x) = \overline{u}$, with 
$\overline{u}[0] = ac(x)$.
We generally drop the subscript on $lbc$, as it is clear from context.

\textbf{Execution Semantics.} 
At step $i$, given system state $x_i$ and DM state $\overline{s_i}$, the next system state $x_{i+1}$ and next DM state $\overline{s_{i+1}}$ are computed with the following sequence of steps: 
\um{(1)~$z_i = ac(x_i)$};
(2)~$\overline{t_i} = lbc(x_i)$, with $\overline{t_i}[0] = z_i$; 
(3)~$\overline{s'_i} = \dmupdate(x_i, \overline{s_i}, \overline{t_i})$; 
(4)~$(u_i, \overline{s_{i+1}}) = \dmstep(\overline{s'_i})$; 
(5)~$x_{i+1} = f(x_i, u_i, w_i)$, for some disturbance $w_i \in \mathcal{W}$.

\subsection{Safety and Transparency Theorems}
\label{sec:theorems}

We define several relevant concepts and then state and prove safety and transparency theorems for Black-Box Simplex. 
\begin{definition}[Safe System Execution]
A system execution is called \emph{safe} if and only if the system state is admissible at every step.
\end{definition}
Safety can be ensured by following a permanently safe command sequence from a given system state.
\begin{definition}[Permanently Safe Command Sequence]
Given state $x_i$, a \um{k-length} \emph{permanently safe command sequence} $\overline{s_i} \um{\in \mathcal{U}^k}$ is one where the state $x_j$ is admissible at every step $j \geq i$, where $(u_i,\overline{s_{i+1}}) = \dmstep(\overline{s_i})$, and $x_{i+1} = f(x_i, u_i, w_i)$, for every choice of disturbance $w_i \in \mathcal{W}$.
\end{definition}
That is, the system state will remain admissible when applying each command in the sequence $\overline{s_i}$, and then repeatedly using the last command forever, according to the semantics of $\dmstep$. 
\um{More general definitions of permanently safe command sequences could be considered, such as repeating a suffix rather than just the last command. For simplicity we do not explore this here.
}

We define recoverable commands to be commands  that result in states that have permanently safe command sequences.
\begin{definition}[Recoverable Command]
Given state $x_i$, a \emph{recoverable command} $u$ is one where there exists a permanently safe command sequence from $x_{i+1}$, where $x_{i+1} = f(x_i, u, w_i)$, for every choice of disturbance $w_i \in \mathcal{W}$.
\end{definition}
Optimal decision modules are defined by requiring the \dmupdate function accept all sequences that can guarantee future safety.
\begin{definition}[Optimal Decision Module]
An \emph{optimal decision module} has a \dmupdate function that accepts $\overline{t}$ at state $x$ if and only if $\overline{t}$ is a permanently safe command sequence starting from $x$.
\end{definition}
A correct DM is one which only accepts sequences that can guarantee future safety. A correct DM, by this definition, could reject every command sequence.
\begin{definition}[Correct Decision Module]
\label{def:correct_dm}
A \emph{correct decision module} has a \dmupdate function that accepts $\overline{t}$ at state $x$ only if $\overline{t}$ is a permanently safe command sequence starting from $x$.
\end{definition}
The role of the BC is to try to keep the system safe. 
An optimal look-ahead BC can be defined as one that always produces a permanently safe command sequence when it exists.
This is optimal in the sense that during system execution, it allows the DM to override the AC as infrequently 
as possible while still guaranteeing safety.
This notion of optimality can be defined with respect to a specific advanced controller $ac$.
%
%
\begin{definition}[Optimal Look-Ahead Baseline Controller]
\um{Given state $x$ with $u = ac(x)$, if there exists a permanently safe command sequence $\overline{s}$ from $x$ with $\overline{s}[0] = u$, then an \emph{optimal look-ahead
baseline controller} 
will always produce a permanently safe command sequence $\overline{t}$, with $\overline{t}[0] = u$.}
\end{definition}
Note that $\overline{t}$ may differ from $\overline{s}$, as there can be multiple permanently safe command sequences from the same state.
%
%
\begin{theorem}[\textbf{Safety}]
\label{thm:safety}
Given initial state $x_0$ along with an initial permanently safe command sequence $\overline{s_0}$, if the decision module is correct, then the system's execution is safe regardless of the outputs of the advanced controller $ac$ and look-ahead baseline controller $lbc$.
\end{theorem}
\begin{proof}
The command executed at each step comes from the state of the decision module $\overline{s_i}$, which maintains the invariant that $\overline{s_i}$ is always a permanently safe command sequence from the current system state $\overline{x_i}$. 
The \dmupdate function can only replace a permanently safe command sequence with another permanently safe command sequence.
Since initially, $\overline{s_0}$ is permanently safe, then by induction on the step number, the decision module's command sequence at every step is permanently safe, and so the system's execution is safe.
\end{proof}

Although safety is important, achieving only safety is trivial, as a decision module can simply reject all new command sequences.
A runtime assurance system must also have a transparency property, where the advanced controller retains control in sufficiently well-designed systems.

\begin{theorem}[\textbf{Transparency}]
\label{thm:transparency}
\um{If (i)~from every state $x_i$ encountered, the output of the advanced controller $ac(x_i) = z_i$ is a recoverable command, (ii)~the look-ahead baseline controller is optimal, and (iii)~the decision module is optimal, then the input command used to actuate the system at every step is the advanced controller's command, $z_i$.}
\end{theorem}
\begin{proof}
The proof proceeds by stepping through an arbitrary step $i$ of the execution semantics defined in Section~\ref{sec:bbdefs}.
Since the output of the advanced controller $ac(x_i) = z_i$ is assumed to be recoverable, there exists a permanently safe command sequence from $x_i$ that starts with $z_i$. 
By the definition of an optimal look-ahead baseline controller, since there exists a permanently safe command sequence, the output $lbc(x_i) = \overline{t}$ must also be a permanently safe command sequence, with $\overline{t}[0] = z_i$ as required by the definition of a look-ahead baseline controller.
In step~(3) of the execution semantics, $\dmupdate(x_i, \overline{s_i}, \overline{t_i}) = \overline{s'_i}$.
Since $\overline{t}$ is a permanently safe command sequence and the decision module is optimal, the command sequence will be accepted by the decision module, and so $\overline{s'_i} = \overline{t}$.
Step~(4) of the execution semantics produces $u_i$, which is the first command in the sequence $\overline{t}$. 
As shown before, this command is equal to $z_i$, which is used in step~(5) of the execution semantics to actuate the system.
This reasoning applies at every step, and so the advanced controller's command is always used.
\end{proof}

\noindent
\textbf{Discussion.}
There are several practical considerations with the described approach.
For example, the black-box controllers may not only generate unsafe commands, but a controller implementation may fail to generate a command at all, for example, entering an infinite loop.
To account for such behaviors, a runtime cap can be used with a default command sequence assumed if the DM receives no input.
For increased protection, the black-box controllers can be isolated on dedicated hardware~\cite{bak2009rtas} so that they do not, for example, crash a shared operating system.
Also, the DM's analysis of the command sequence is nontrivial and could involve a runtime reachability computation.
If this may take too long, we again could use a runtime cap. 
This means that the practicality of the architecture depends on the efficiency of runtime reachability methods, an active area of research orthogonal to this work.

Another consideration is the feasibility of coming up with permanently safe command sequences.
For systems where landing or coming to a stop is considered safe, remaining there forever will be permanently safe.
Other approaches, which we use the case studies in the next section, rely on geometric arguments to show permanent safety. 
Methods from control theory could also be used for this, such as computing forward invariant sets~\cite{kapinski2015discovering} or using a locally stable controller.
For example, using the indirect method of Lyapnuov, a closed-loop system's equilibrium point $x^*$ can be proven to be stable using linearization, along with conservative bounds on its basin of attraction~\cite{murray1994mathematical}.
The BC would then strive to get the system into the basin of attraction of $x^*$, and then use the locally stable controller to ensure indefinite future safety.
Directly using the locally stable controller as the BC, however, would be overly conservative, as it would not allow the system to leave the (potentially small) basin of attraction.

%
%
%

\section{Case Studies}
\label{sec:eval}
In this section, we apply the approach to two case studies: a multi-robot coordination system, and a mid-air collision avoidance system for groups of F-16 aircraft.
%
%
\subsection{Multi-Robot Coordination}
\label{ssec:robot}

We consider a multi-agent system (MAS), indexed by $\mathcal{M}=\{1,...,n\}$, of planar robots modeled with discrete-time 
dynamics of the form:
\begin{equation}
\begin{split}
    p_i(k+1) &= p_i(k) + dt \cdot v_i(k),\quad |v_i(k)| < v_{max} \\
    v_i(k+1) &= v_i(k) + dt \cdot a_i(k),\quad |a_i(k)| < a_{max}
\end{split}
\end{equation}
where $p_i$, $v_i$, $a_i$ $\in \mathbb{R}^2$ are the position, velocity and acceleration of agent $i$, respectively, at time step $k$, and $dt \in \mathbb{R}^+$ is the time step.
The magnitudes of velocities and accelerations are bounded by $v_{max}$ and $a_{max}$, respectively.
The acceleration $a_i$ is the control input for agent $i$. The combined state of all agents is denoted as $x = [p_1^T, v_1^T,...,p_n^T, v_n^T]^T$, and their  accelerations are $a = [a_1^T,...,a_n^T]^T$.


\begin{figure}[t]
\centering
\subfloat[Initial configuration, ${k = 1}$]{\includegraphics[width=.25\columnwidth]{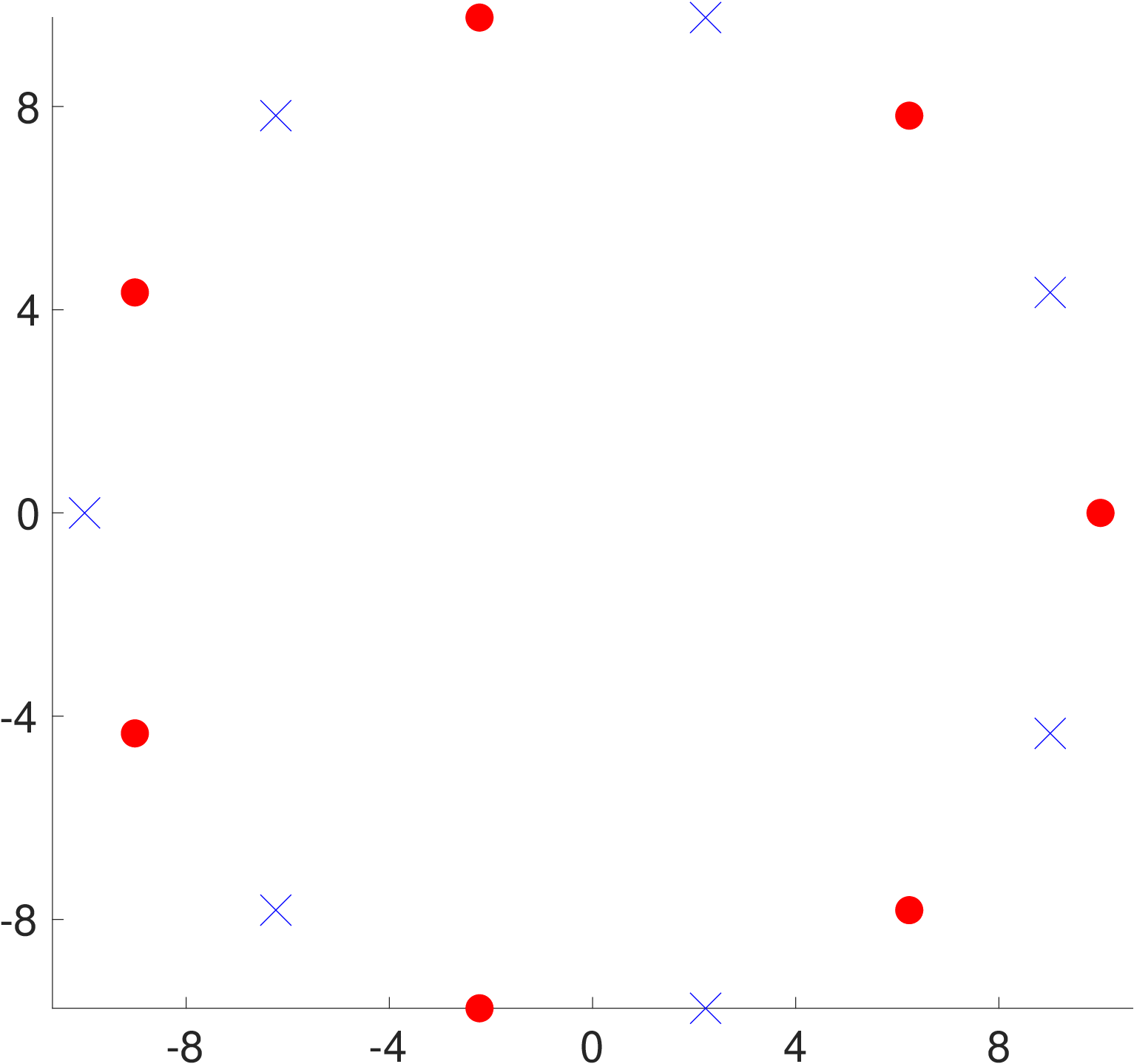}}\hfill
\subfloat[$k = 10$]{\includegraphics[width=.25\columnwidth]{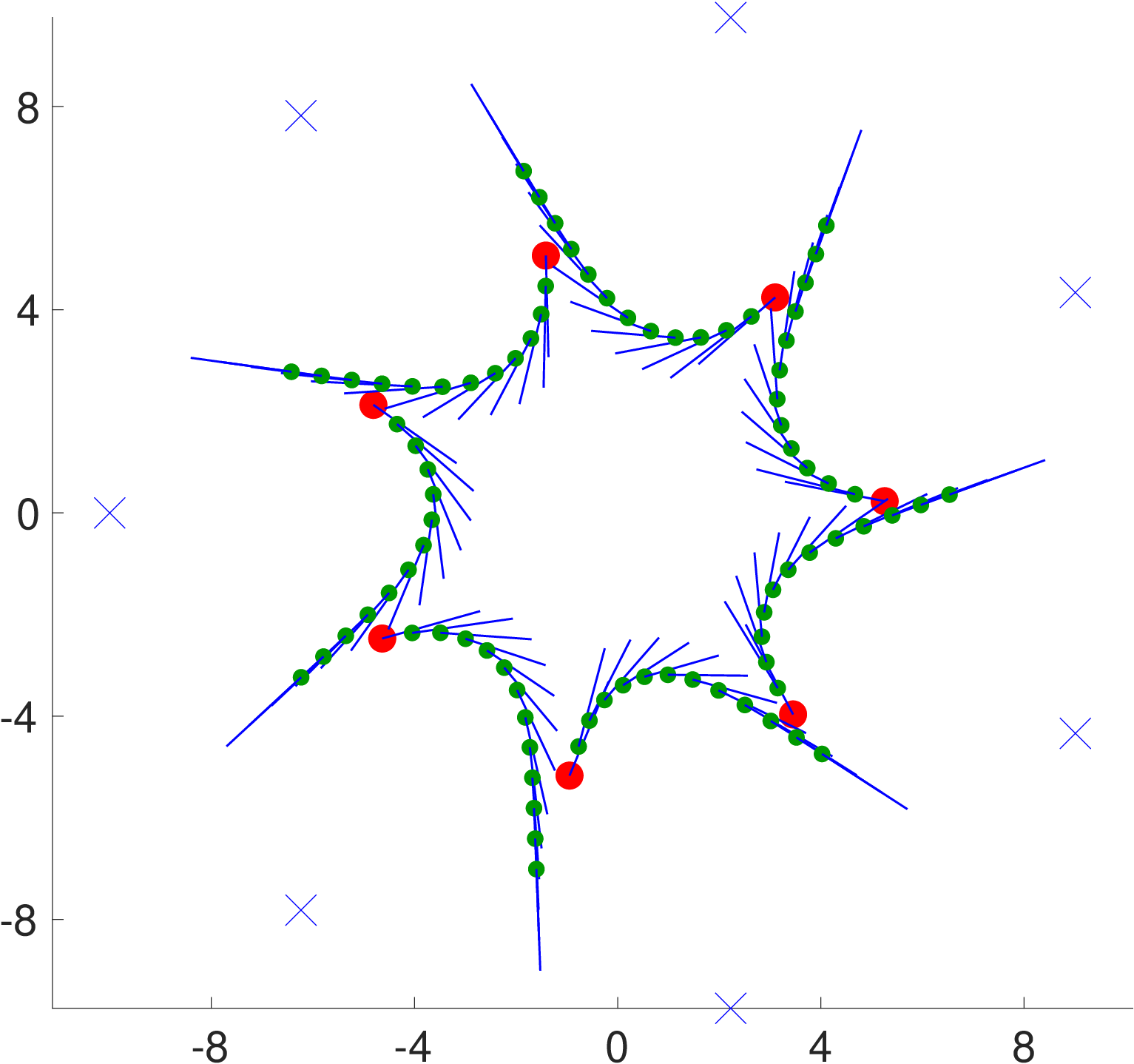}}\hfill
\subfloat[BC fails, $k = 11$]{\includegraphics[width=.25\columnwidth]{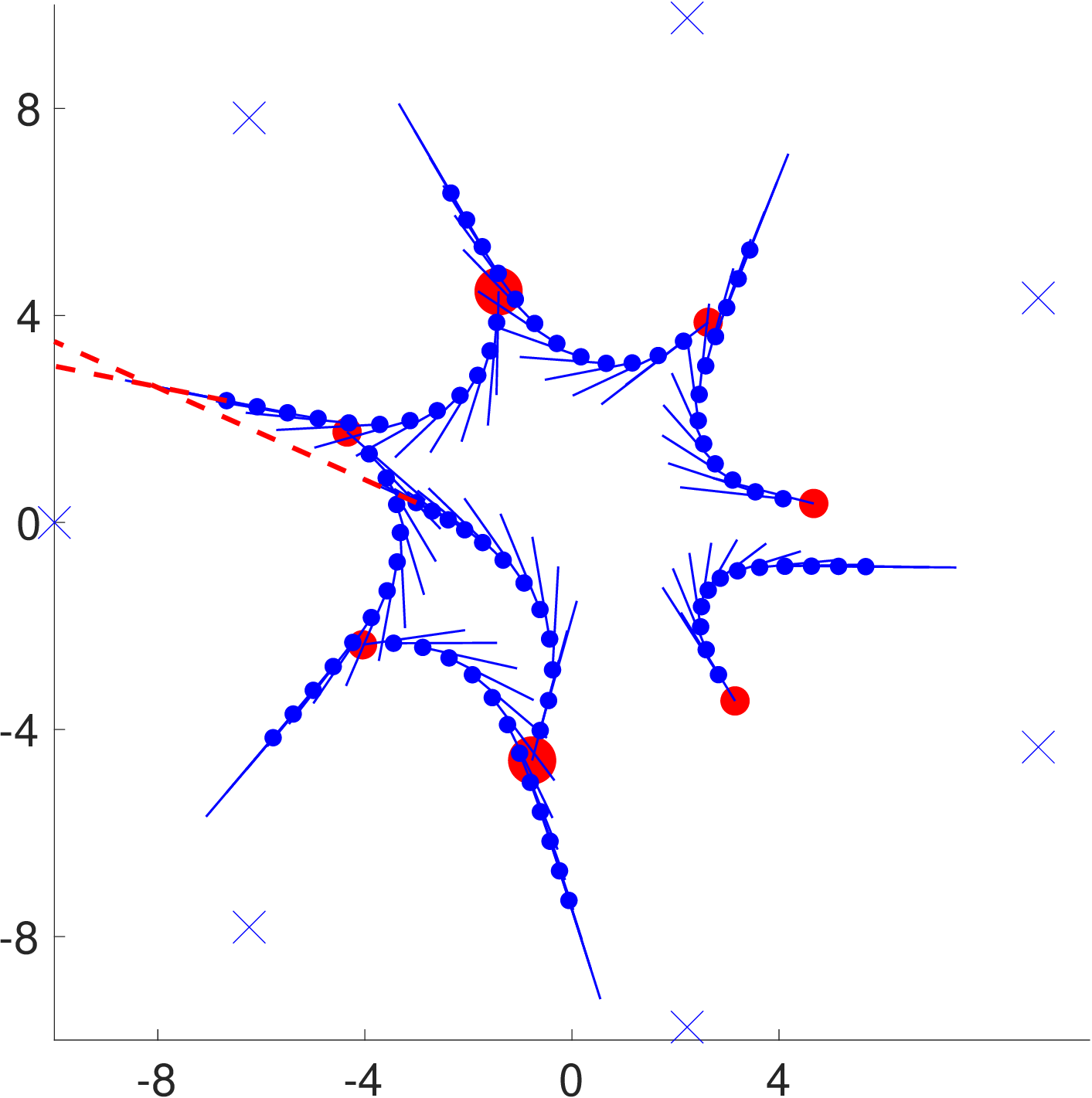}} \hfill
\subfloat[Final configuration, $k = 32$]{\includegraphics[width=.25\columnwidth]{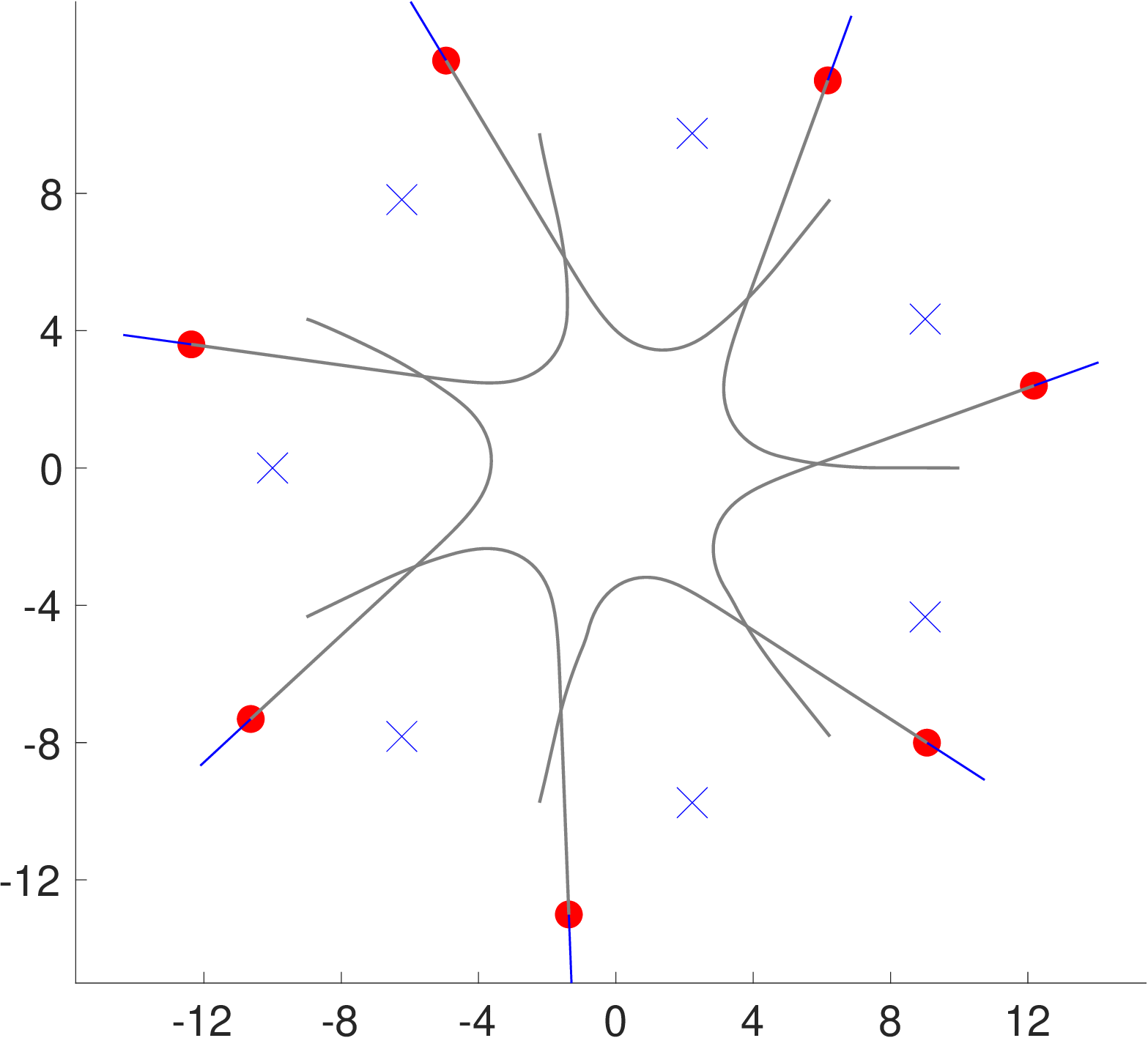}}\hfill
\caption{Simulation of the MAS with 7 robots. The DM performs system recovery after the BC produces an unsafe command sequence. 
The BC's proposed path is shown in part (c) at ${k = 11}$, where the two dotted red lines intersect, indicating the future paths of the agents cross. 
We represent current positions as red dots, future positions corresponding to the safe/unsafe command sequences as green/blue dots, velocities as blue lines, and agent trajectories as grey curves. 
}
\label{fig:bc_fail} \vspace{-2em}
\end{figure}

In the initial configuration, the agents are equally spaced on the boundary of a circle and are at rest. Agent $i$'s goal is to reach a target location $r_i$, located on the opposite side of the circle. The initial configuration of the MAS is shown in Figure~\ref{fig:bc_fail}(a), where the agents and their target locations are represented as red dots and blue crosses, respectively.
The safety property is absence of inter-agent collisions. A pair of agents is considered to collide if the Euclidean distance between them is less than a non-negative threshold  $d_{min}$. Thus, the safety property is that ${\left \| p_i - p_j \right \| > d_{min}}$ for all pairs of agents $i,j \in \mathcal{M}$ with $i \ne j$. 

Both the AC and the BC are designed using centralized Model Predictive Control (MPC), which produces command sequences as part of the solution of a nonlinear optimization problem.
For collision avoidance, we use a potential field formulation \cite{khatib1986real} in both the AC and BC.
%
While the AC tries to reach the target positions on the opposite side of the circle, the BC has a simpler goal of having each agent leave the circle.
Note that numerical methods for global nonlinear optimization, such as MATLAB's \texttt{fmincon} used in our implementation, do not provide a guaranteed optimal solution.
To create unsafe variants of the controllers, we simply limit the number of iterations used for optimization.
\um{The AC only outputs the first command of the command sequence, whereas the BC produces the full command sequence.} 
%
%
\um{Both the AC and the BC are high-level controllers that produce accelerations. In our simulations, we do not model the low-level controller; the plant dynamics 
work directly with the accelerations. When implementing our approach on physical robots, 
a trusted low-level controller will map the desired acceleration commands to actuator inputs.}

A centralized MPC controller produces a command sequence $\overline{s}$ of length $T$, where $T$ is the prediction horizon, and each command $\overline{s}[i]$ contains the accelerations for all agents to use at step $i$.

The centralized MPC controller solves the following optimization problem at each time~step~k: 
\begin{align}
\label{eq:cen_mpc}
\underset{a(k \mid k), \ldots, a(k+T-1 \mid k)}{\arg\min} \; &\sum_{t=0}^{T-1} J(k + t\mid k) 
\; + \;
\lambda \cdot \sum_{t=0}^{T-1} \| {a(k + t\mid k)}\|^2
\end{align}
where $a(k+t \mid k)$ and $J(k+t\mid k)$ are the predictions made at time step $k$ for the values at time step $k + t$ of the accelerations and the centralized (global) cost function $J$, respectively.
The first term is the sum of the centralized cost function, evaluated for $T$ time steps, starting at time step $k$.  It encodes the control objective.
The second term, scaled by a weight $\lambda > 0$, penalizes large control inputs.

\textbf{Advanced controller.}
The centralized cost function $J_{ac}$ for the AC contains  two  terms: (1)~a \emph{separation} term based on the inverse of the squared distance between each pair of agents (potential field term for collision avoidance); and (2)~a \emph{target seeking} term based on the distance between the agent and its target location.  
\begin{equation}
\label{eq:mas_ac}
    J_{ac} = \omega_s \sum_{i>j}\frac{1}{\left \| p_i - p_j \right \|^2} + \omega_t \sum_i \left \| p_i - r_i \right \|^2
\end{equation}
where $\omega_s, \omega_t \in \mathbb{R}$ are the weights of the separation term and target seeking terms.
The separation term promotes inter-agent spacing but does not guarantee collision avoidance. 
The AC generates a command sequence by solving the optimization problem in Eq.~\ref{eq:cen_mpc}, with $J$ replaced by $J_{ac}$.
The first command in that sequence is the AC's command; it is passed to the LBC.

\textbf{Baseline controller.}
The centralized cost function $J_{bc}$ for the BC contains two  terms.
As in Eq.~\ref{eq:mas_ac}, the first term is the separation term (collision avoidance based on potential fields). 
The second term is a \emph{divergence} term which forces the agents to move out of the circle by aligning their velocities with rays radially pointing out of the center of the circle.  
  
\begin{equation}
\label{eq:mas_bc}
    J_{bc} = \omega_s \sum_{i>j}\frac{1}{\left \| p_i - p_j \right \|^2} + \omega_d \sum_i \left ( 1 - \frac{\left (p_i - c \right ) \cdot v_i}{|p_i - c||v_i|} \right )
\end{equation}
where $\omega_s, \omega_d \in \mathbb{R}$ are the weights of the separation term and the divergence term, and $c$ is the center of the circle containing the initial configuration of the robots and their target locations. The control law for the BC is Eq.~\ref{eq:cen_mpc}, with $J$ replaced by $J_{bc}$. A zero acceleration is appended to the end of the BC's command sequence to help establish collision freedom for all future time steps. 

\textbf{Decision module.} 
The LBC combines accelerations from the AC and the BC,
producing the command sequence ${\overline{t} = [ac(x), bc(x'), \vec{0}]}$, where $x'$ is the next state after executing $ac(x)$ in state $x$.  
The function $\dmupdate(x, \overline{s}, \overline{t})$ accepts the proposed command sequence $\overline{t}$ if and only if $\overline{t}$ is a permanently safe command sequence.
For this system, a command sequence $\overline{ t}$ is considered permanently safe in a state $x$ if it satisfies the following two conditions. First, for all states in the state trajectory obtained by executing $\overline{t}$ from $x$, the Euclidean distance between every pair of distinct agents is at least $d_{min}$. Second, in the final state, for all pairs of distinct agents, the rays extending from their positions and in the directions of their velocities do not intersect.  
\um{ Any pair of 
agents that satisfies the second condition will not collide in the future, since the last command in the sequence $\overline{t}$ has zero acceleration.}
The initial permanently safe command sequence is a zero acceleration for all agents, as the agents start at rest. 

\textbf{MPC Parameters.}
In our case study, we use the following MPC parameters: $dt = 0.3~sec$, $d_{min} = 1.7$, $a_{max} = 1.5$, and $v_{max} = 2$. 
The length of the prediction horizon for MPC is ${T_{ac} = T_{bc} = 10}$.

\textbf{Successful Recovery After Failure.}
We first consider seven robotic agents initialized on a circle centered at the origin, with a radius of~10. 
The initial state of the system is shown in Figure~\ref{fig:bc_fail}(a). At ${k = 11}$, the BC produces an unsafe command sequence.  The state trajectory corresponding to the unsafe sequence is shown in blue.  As shown in Figure~\ref{fig:bc_fail}(c), the final paths of the two agents corresponding to the larger red dots cross after simulating the current state forward with the unsafe sequence.
Hence, at ${k = 11}$, the DM rejects the proposed command sequence and shifts control to the previous safe command sequence, which safely recovers the system. 
%
%
Here, we purposefully did not return control to the AC to demonstrate how the stored command sequence keeps the agents safe~\footnote{A video of the simulation 
is available at \url{https://youtu.be/bcVJBkGgnxA}.}.

\begin{figure}[t]
\centering
\includegraphics[width=.38\columnwidth]{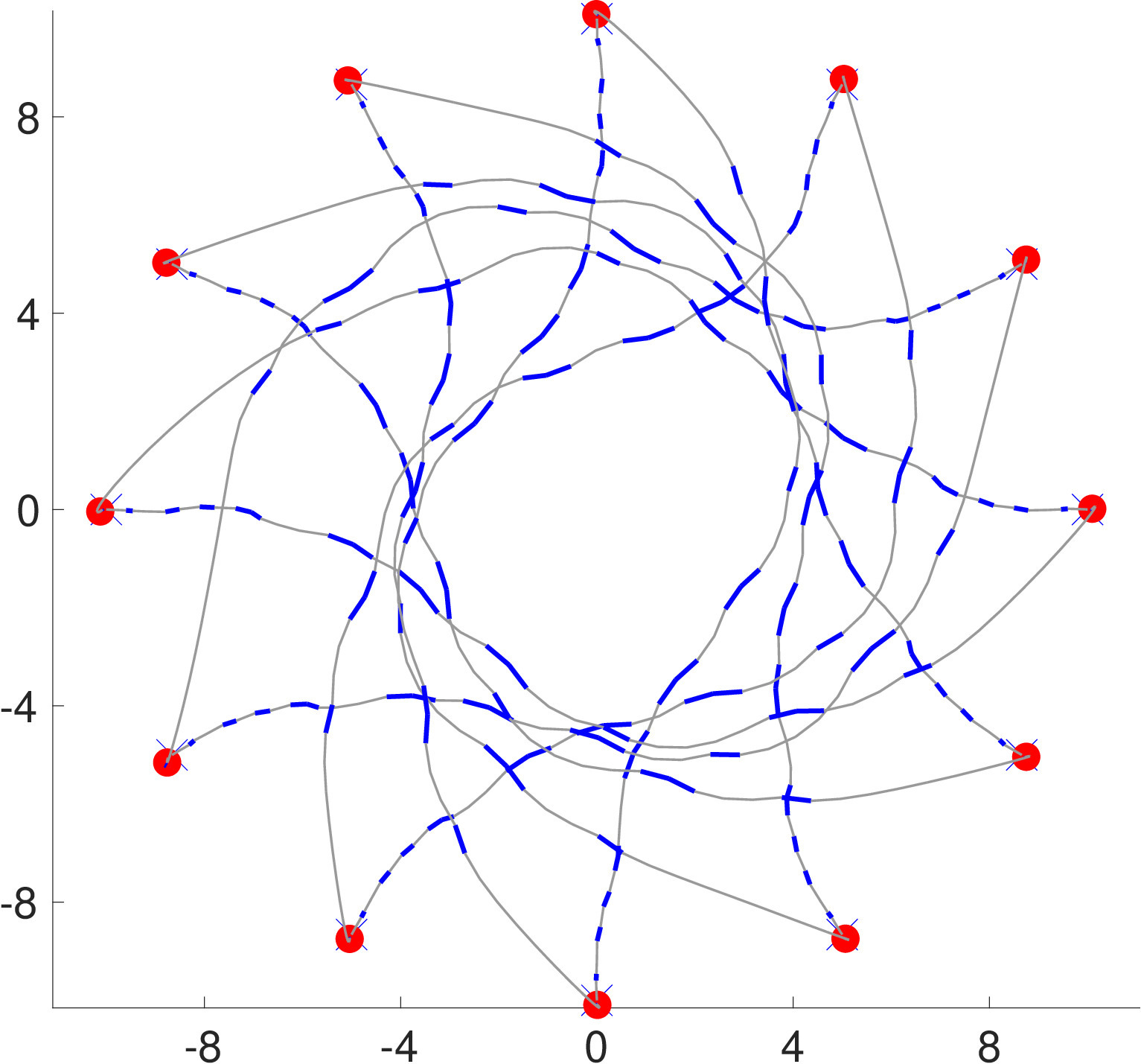}
\caption{Stress test of robotic MAS with 12 robots reaching their targets. 
Trajectory segments where stored command sequences are used are shown in blue.
} \label{fig:ac_success} \vspace{-2em}
\end{figure}

\textbf{Reverse Switching Scenario.}
We stress-tested the multi-robot system by initializing 12 agents on a circle of radius~$10$. The path of the agents is shown in Figure~\ref{fig:ac_success}. 
\um{There are 10 instances where the DM rejects the AC's proposed command sequence and instead uses the stored command sequence.}
Nonetheless, all agents reach their target locations without colliding, maintaining a minimum separation of 1.724 between any pair of agents\footnote{A video of the simulation is available at \url{https://youtu.be/qmk31jS6B2Y}.}.





\textbf{Handling Uncertainty.}
%
%
We next investigate the DM's runtime overhead when there is uncertainty in the robot's state or the dynamics.
The former case arises when the sensors used to determine the positions and velocities are  subject to sensor noise.
The latter case could be used to account for modeling errors, through disturbances on the positions and velocities at each step.

We continue to use the same MPC strategy as before; thus, the controllers ignore the uncertainty when generating proposed command sequences.  Only the logic used by the DM to accept or reject command sequences is modified  to account for uncertainty.
We examine the scenario shown before in Figure~\ref{fig:bc_fail}(b).
To account for the uncertainty, we perform an online reachability computation.
To do this, we use efficient methods for reachability for linear systems based on zonotopes~\cite{girard05}, which we implement in Python. 
Briefly, a zonotope is a set of states represented as an affine transformation of a unit box.
The unit box is associated with a number of \emph{generator vectors}, where each generator vector corresponds to one dimension of the box.
The computational efficiency of propagating sets over time using zonotopes relates to the number of generators.
Each agent has four state variables, two for position and two for velocity. %
The composed system with seven agents has 28 state variables.

In the situation shown in Figure~\ref{fig:reach}(a), the current state is assumed to have uncertainty independently in both position and velocity with an $L^2$ norm of 0.1. We use a 16-sided polygon to bound this uncertainty.
In the plot, the deterministic simulation is given, along with black polygons for each agent that show the states that might be reachable at each step due to the sensor uncertainty.
The uncertainty in the velocity causes the set to expand over time, since the open-loop command sequence does not attempt to compensate for the uncertainty.
%
The zonotope representation of the composed system needs 112 generator vectors to represent the initial states, which remains constant at every time step.

In the situation shown in Figure~\ref{fig:reach}(b), the initial state has very little error, but the dynamics is modified to have disturbances at each step.
For each component of each agent's position and velocity, we allow an external disturbance value to be added in the range $[-0.02, 0.02]$.
Since each agent has four independent disturbances, the zonotope representation of the composition will have 28 new generators added at each step.
After 12 steps, the final zonotope will have a total of 364 generators.

\textbf{Runtime.} To measure runtime, we used a standard laptop with a 2.70 GHz Intel Xeon E-2176M CPU and 32 GB RAM.
The method is
fast.
For the case of sensor uncertainty, computing the box bounds of the reachable set at all the steps takes about 1.5 milliseconds.
With uncertainty, even though the number of generators grows over time, it is not large enough to significantly affect the runtime.
The computation with disturbances requires about 2 milliseconds to complete.
%
We believe such execution times are sufficiently fast for use in the decision module.

\begin{figure}[t]
\centering
\subfloat[Reachable States with Sensor Error]
{\includegraphics[width=.48\columnwidth]{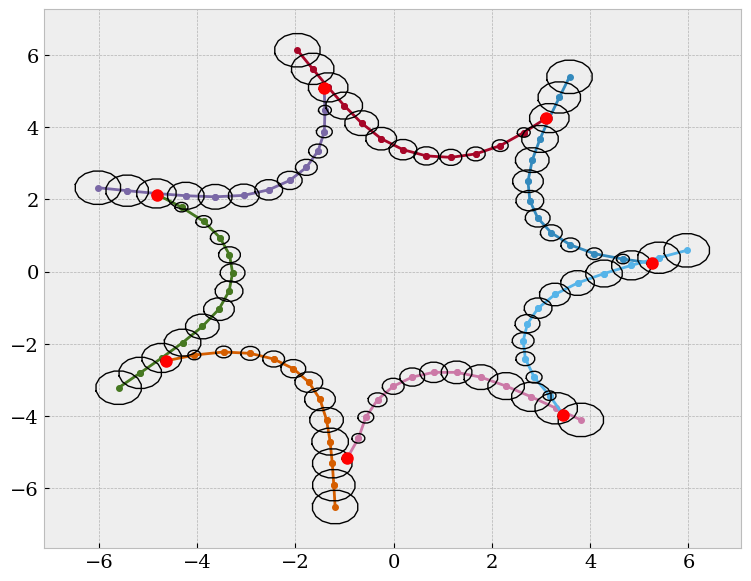}}\hfill
\subfloat[Reachable States with Disturbances]
{\includegraphics[width=.48\columnwidth]{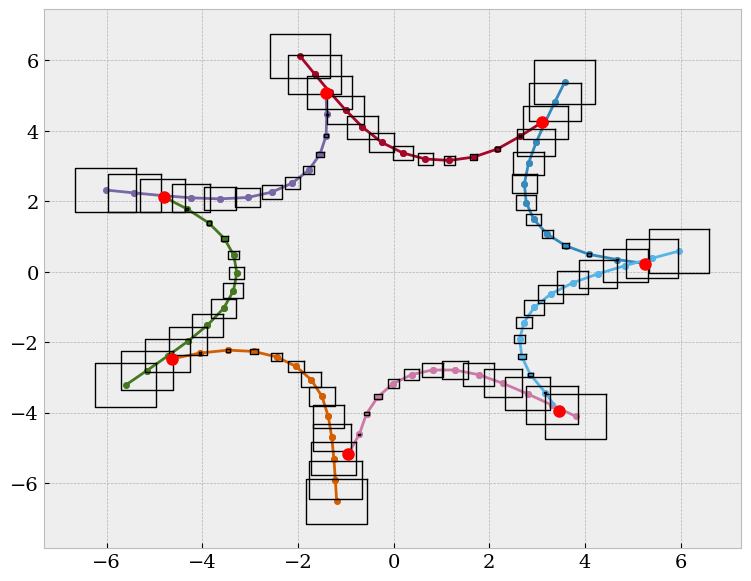}}\hfill
\caption{Zonotope reachability computes future states with uncertainty.}\label{fig:reach} \vspace{-2em}
\end{figure}

%
%
%
%
%
%
%

\subsection{Multi-Aircraft Collision Avoidance}

Our second evaluation system guarantees collision avoidance for groups of aircraft.
We use a full six-degrees-of-freedom F-16 simulation model~\cite{heidlauf2018arch}, based on dynamics taken from an Aerospace Engineering textbook~\cite{stevens2015aircraft}. 
Each aircraft is modeled with 16 state variables, including positional states, positional velocities, rotational states, rotational velocities, an engine thrust lag term, and integrator states for the low-level controllers.
These controllers actuate the system using the typical aircraft control surfaces---the ailerons, elevators, and rudder---as well as by setting the engine thrust.
The system evolves continuously with piece-wise nonlinear differential equations, where the function that computes the derivative given the state is provided as Python code.
In order to match the discrete-time plant model in Definition~\ref{eq:control_system}, we periodically select a control strategy with a frequency of once every two seconds.
The model further includes high-level autopilot logic for waypoint following, which we reuse in the advanced controller.

For the collision-avoidance baseline controller, our controller is based on the ACAS Xu system designed for collision avoidance in unmanned aircraft~\cite{kochenderfer2011robust}.
While the original system was designed using a partially observable Markov decision process (POMDP), the resultant controller was encoded in a large look-up table that used hundreds of gigabytes of storage~\cite{julian2019deep}.
To make the system more practical, one early approach considered a downsampling process followed by a lossy compression using neural networks~\cite{katz2017reluplex,julian2019deep}.
We use these downsampled neural networks as the BC and refer to this as the original system.

The system issues horizontal turn advisories based on the relative positions of two aircraft, an \emph{ownship} and an \emph{intruder}.
The system is similar to Simplex, where the output can be either \emph{clear-of-conflict}, where any command is allowed, or an override command that is one of \emph{weak-left}, \emph{weak-right}, \emph{strong-left} or \emph{strong-right}.
We adapt this system to the multi-aircraft case by having each aircraft run an instance of the system against every other aircraft, using the closest turn advisory as the output.
%
%

To create command sequences, the BC repeatedly advances the plant model and re-runs the collision avoidance system in a closed-loop fashion until the generated command sequence is permanently safe.
%
To check whether a generated command sequence is permanently safe, the DM checks that (i)~each aircraft's state stays within the model limits (e.g., no aircraft enters a stall), (ii)~all aircraft obey the safety distance constraint at all times, and (iii)~the execution ends in a state where the roll angle of each aircraft has been small (less than 15 degrees) and the distances between all pairs of aircraft has been increasing consecutively for several seconds.
If all aircraft continue to fly straight and level from such a configuration, their distance would increase and no collisions would occur in the future.

As with the multi-robot scenario, we examine cases where the initial aircraft state $x_0$ has all aircraft starting evenly-spaced, facing towards the center of a circle with a given initial diameter.
Each aircraft has an initial velocity of 807 ft/sec and an initial altitude of 1000 ft, both of which are maintained throughout the maneuver by the lower-level controllers.
The AC commands each aircraft to fly towards a waypoint past the opposite side of the circle, which would cause a collision at the center.
The safety property requires maintaining horizontal separation.
The \emph{near mid-air collision cylinder} (NMAC) uses a safe horizontal separation of 500 ft~\cite{marston2015acas}, although we will vary this in our evaluation.
For the initial permanently safe command sequence $s_0$, we have each aircraft fly in clockwise circles forever, which avoids collisions.
%

In addition to the AC being unsafe, the baseline controller should not be fully trusted for many reasons. 
The original POMDP formulation was not proven formally correct, not to mention the downsampling and lossy neural network compression.
While some research has examined proving open-loop properties for the neural network compression~\cite{katz2017reluplex,bak2021vnncomp,bak2020cav}, these do not imply \emph{closed-loop} collision avoidance.
Further, we use a multi-aircraft adaptation of the system, which could also lead to problems.
Although aspirationally, the system should handle up to 30 intruders~\cite{julian2019deep}, in practice most analysis has been performed on two aircraft scenarios.
Finally, the intended physical system response to the collision-avoidance commands is that \emph{weak-left} and \emph{weak-right} should cause turning at 1.5 degrees per second, whereas \emph{strong-left} and \emph{strong-right} turn at 3.0 degrees per second~\cite{julian2019deep}.
However, turning an aircraft in the F-16 model (as well as in the real world) is not an instantaneous process, and requires first performing a roll maneuver before the heading angle begins to change.
For these reasons, the BC in this scenario is also an unverified component, and we will show scenarios where it misbehaves.
Nonetheless, we will compose the incorrect AC with the incorrect BC to create a safe collision-avoidance system by using \bsa.

We now elaborate on three scenarios: (i)~a three aircraft case, which shows the safety of the system despite unsafe outputs, (ii)~a four aircraft case, which shows the increased transparency of \bsa, and (iii)~a 15 aircraft case, which shows safe navigation of a complex scenario.
Also, a seven aircraft case is presented in Appendix \ref{appendix:seven},
which shows the safety condition can be easily customized.

In all the plots in this section, we show snapshots at the time when the distance between the two closest aircraft is smallest.
The two red aircraft are the closest pair, and their distance is printed in the bottom right of each figure.
The solid line shows the historic path of each aircraft, and the dotted line is the future trajectory.

\textbf{Three Aircraft Scenario.}
The original collision avoidance system was designed with two aircraft in mind, an ownship and an intruder.
We adapted it to the multi-aircraft case, but this mismatch between the system design assumptions and usage scenario can lead to problems.
%
%
%
In Figure~\ref{fig:acasxu3}, we show such a scenario, where the initial circle diameter is 90,000 ft.
In Figure~\ref{fig:acasxu3}(a), the minimum distance between the top two aircraft is 175 ft, violating the near mid-air collision safety distance.
The other two subplots show the system using \bsa with a safety distance of 1500 ft; the minimum separation is 1602 ft, which satisfies the constraint.
 
 \begin{figure*}[t]
\centering
\subfloat[Original System]
{\includegraphics[width=.32\textwidth]{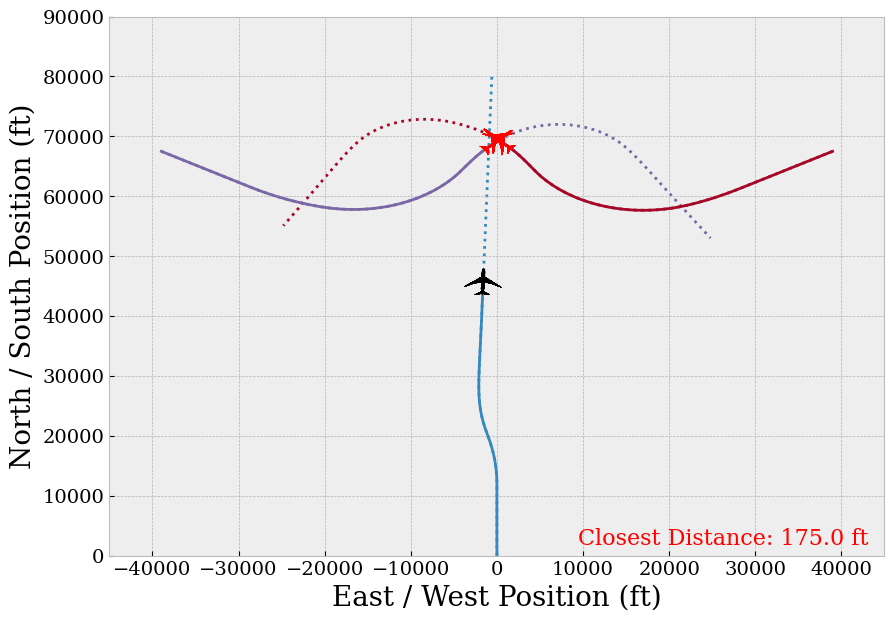}}
\subfloat[Black-Box Simplex]
{\includegraphics[width=.32\textwidth]{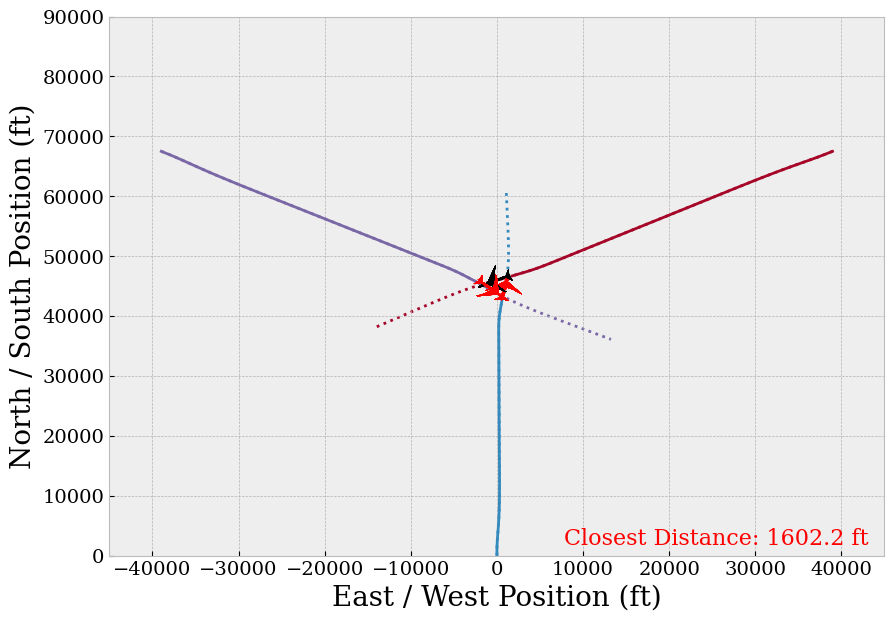}}
\subfloat[Black-Box Simplex (Zoomed In)]
{\includegraphics[width=.32\textwidth]{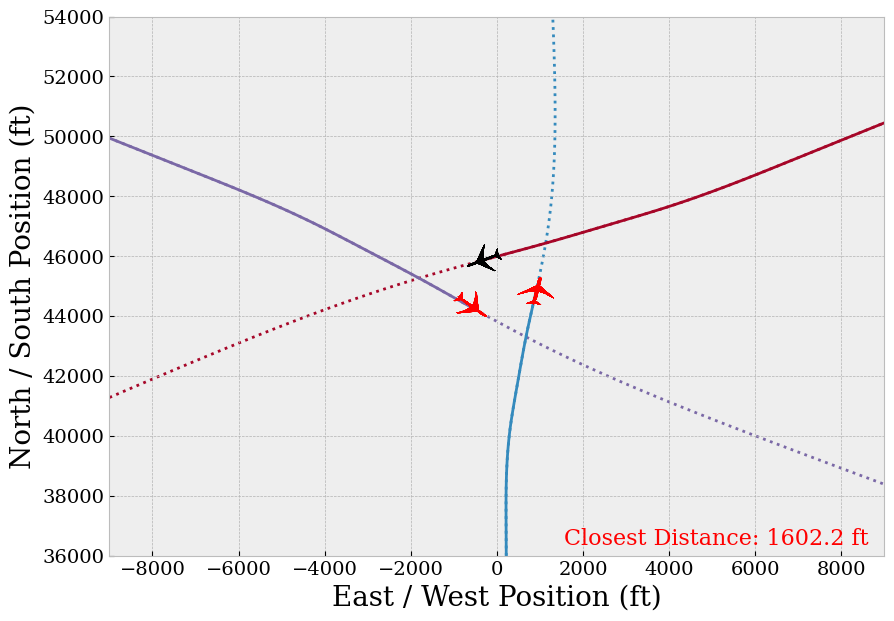}}
\caption{Black-Box Simplex is safe. In the three-aircraft case, the original system fails, whereas \bsa maintains the 1500 ft separation.}\label{fig:acasxu3} \vspace{-2em}
\end{figure*}

\textbf{Four Aircraft Scenario.}
Figure~\ref{fig:acasxu4} shows a four-aircraft scenario using an initial circle diameter of 70,000 ft.
In this case, both designs have safe executions. 
Using the original system leads to a minimum separation of 5342 ft, whereas the minimum separation with Black-Box Simplex is 1600 ft, much closer to the 1500 ft safety-distance constraint used in the DM.
Although both systems are safe, from the plots it is clear that the Black-Box Simplex version is more transparent, in the sense that it produces smaller modifications to the direct-line trajectories commanded by the AC.



\textbf{Fifteen Aircraft Scenario.}
Finally, we demonstrate the system's ability to safely navigate complex scenarios.
For this, we use a 15 aircraft scenario, with an initial circle diameter of 90,000 ft.
With 15 aircraft, the composed system has 240 real-valued state variables, each of which evolves according to piece-wise nonlinear differential equations.
The plot for this system was shown in the introduction in Figure~\ref{fig:acasxu15}.
While the original system is unsafe, Black-Box Simplex has a minimum separation of 1500.5 ft, just above the 1500 ft safety constraint used in the DM.
Another surprising observation is that in some of the cases, such as this 15-aircraft case and the seven-aircraft case shown in the appendix of the extended report\footnotemark[\value{footnote}], the aircraft perform something similar to a roundabout maneuver.  This is an emergent behavior, not something explicitly hardcoded or anticipated.
A video of this case is also available online\footnote{\url{https://youtu.be/Bhn0uqKCj7Q}}.

\textbf{Runtime.} The existing implementation uses numerical integration for the dynamics with an adaptive-step explicit Runge-Kutta scheme of order 5(4) from Python's \texttt{scipy} package.
On our laptop platform with default accuracy parameters, this runs at about 55 times faster than real-time per aircraft.
%

\begin{figure}[t]
\centering
\subfloat[Original System]
{\includegraphics[width=.48\columnwidth]{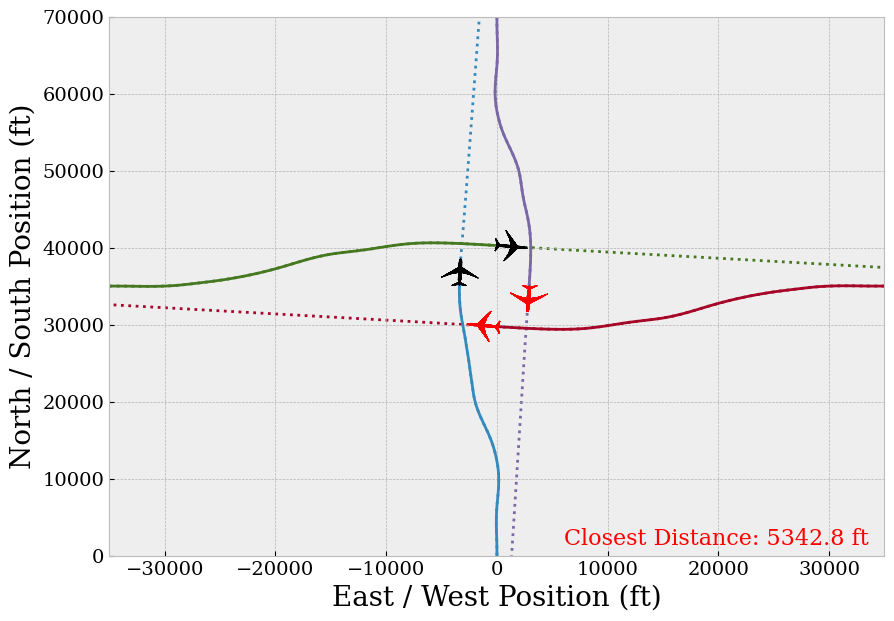}} ~~~
\subfloat[Black-Box Simplex]
{\includegraphics[width=.48\columnwidth]{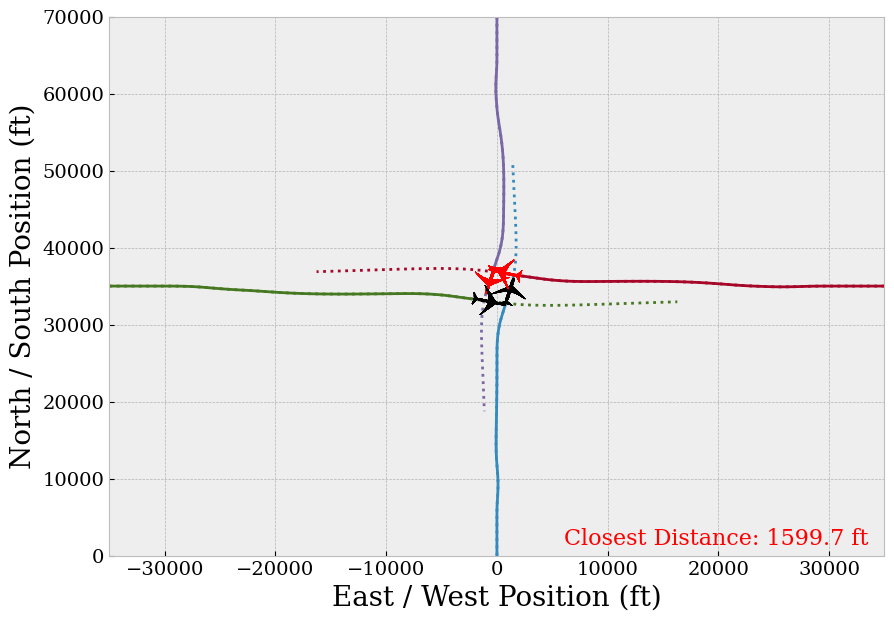}}\hfill
\caption{Black-Box Simplex is more transparent. For the four aircraft case, the original system is significantly more intrusive than Black-Box Simplex, which overrides commands just enough to guarantee the 1500 ft separation requirement. }\label{fig:acasxu4} \vspace{-2em}
\end{figure}

\section{Related Work}
\label{sec:related}

Reachability-based verification methods for black-box systems for waypoint following with uncertainty have been recently investigated in the ReachFlow framework~\cite{reachflow2020}.
%
%
ReachFlow builds upon the Flow* reachability tool~\cite{chen2013flow}, which is unlikely to scale to systems like the 240-variable 15-aircraft scenario. 

\um{A framework for safe trajectory planning using MILP for piecewise-linear vehicle models is presented in~\cite{Schouwenaars2005,schouwenaars2006}. 
The method relies on the ability of an MPC controller to produce command sequences where the terminal state in the prediction horizon is constrained to lie within a safe invariant set.  This provides a safe back-up command sequence for the next step in case the system fails to find a safe sequence. 
%
%
The scope of this work is limited to MPC, and it is not clear how to extend it to other types of controllers.
Moreover, the conditions for switching back from the stored return trajectory are not formalized. }

In the Contingency Model Predictive Control framework~\cite{alsterda2019}, an MPC controller maintains a contingency plan in addition to the nominal or desired plan to ensure safety during an identified potential emergency.
%
Like \bsa, the initial command is common to both plans.
%
%
In this framework, both plans must be generated using their custom version of MPC, whereas Black-Box Simplex works with independent baseline and advanced controllers of any design.

Similar frameworks have been considered for autonomous vehicles, using fail-safe backup plans and reachability analysis~\cite{magdici2016fail}. 
In this case, the target was planning for autonomous vehicles where most likely trajectories are used for other vehicles but safety can still be provided if emergency maneuvers are performed instead.
Other ideas such as Safety Net Control~\cite{schurmann2021formal} extend the approach to use backreachability and underapproximations of nonlinear reachable sets while taking computation time into account.

Designing safe switching logic for a given baseline controller is related to the concept of computing viability kernels~\cite{saint1994approximation} (closed controlled-invariant subsets) in control theory.
This often requires set operations which can be inefficient in high-dimensional spaces with nonlinear dynamics, although there has been some progress on this~\cite{kaynama2012computing,maidens2013lagrangian}.

Simplex designs have also been considered that use a combination of offline analysis with online reachability~\cite{bak2014rtss}.
Again, though, reachability computation is currently intractable for large nonlinear systems, and requires symbolic differential equations.
Other work has used Simplex to provide safety guarantees for neural network controllers with online retraining~\cite{phan2020}.
In these approaches, the baseline controller must be verified ahead of time.

Online simulation-based methods have also been investigated to secure power grids from insider attacks~\cite{mashima2018securing}.
As with this work, fast online simulation is critical, although the goal there is system security not safe high-level control design.

The design of the MPC controllers for our multi-robot case study is similar to Control Barrier Function methods~\cite{borrmann2015,gurriet2018online} and Implicit Active Set Invariance Filtering~\cite{gurriet2019scalable}. 
There, a runtime assurance system was used to provide minimally perturbed advanced controller commands, computed using a constrained-optimization problem.
However, the optimization problem might become infeasible or global nonlinear optimization could perform poorly at one of the steps at runtime, causing this method to be unsafe.
With Black-Box Simplex, failure of the baseline controller does not compromise safety.

%
%

%
%
%

\section{Conclusions}
\label{sec:conclusions}

We have presented the Black-Box Simplex Architecture, a methodology for constructing safe CPS from unverified black-box high-level controllers.
Unlike the classical Simplex design, the baseline controller does not  need to be statically verified and can even be incorrect.
The tradeoff is that the decision module performs more extensive runtime checking and stores backup command sequences produced by the black-box baseline controller at previous time steps. The complexity of runtime checking depends on the nature of the system model. For deterministic models, simulation suffices. However, if the model has uncertainty then we need to perform online reachability analysis.

\bsa reduces the difficult problem of proving high-level safety to a simpler problem of \emph{performance optimization}: ensuring that the runtime checking completes before a  decision is needed.
The practicality of the approach was demonstrated through two significant case studies, including a mid-air collision avoidance system for groups of F-16 aircraft created from imperfect logic encoded in neural networks. This case study involves a highly complex nonlinear system with  over a hundred dimensional variables and a neural-network-based controller.
Black-Box Simplex provides a feasible path for runtime verification of systems that are otherwise unverifiable in practice.

\vspace{1em}
\noindent
\textbf{Acknowledgement.} This material is based upon work supported by National Science Foundation (NSF) under grant numbers OIA-2134840, OIA-2040599, CCF-1918225, CCF-1954837 and CPS-1446832, the Office of Naval Research (ONR) under grants N000142112719 and N000142212156, and the Air Force Office of Scientific Research (AFOSR) under award numbers FA9550-19-1-0288, FA9550-21-1-0121, FA9550-22-1-0450. 
Any opinions, findings, and conclusions or recommendations expressed in this material are those of the author(s) and do not necessarily reflect the views of the NSF, United States Air Force or the United States Navy. 
An early version of this work was presented in the CAADCPS 2021 workshop under the title ``Safe CPS from Unsafe Controllers''~\cite{mehmood2021safe}.

\bibliographystyle{splncs04}
\bibliography{References}

\clearpage
\appendix

\section{Seven Aircraft Case - Safety Condition Customization}
\label{appendix:seven}

We investigated a seven aircraft scenario with an initial circle diameter of 70,000 ft.
Here, the original system violates the horizontal separation constraint, and the minimum separation distance is 277 ft.
We run Black-Box Simplex on this system using three different safety distances, 1500 ft, 1000 ft, and 500 ft.
All avoid collisions, and as the safety distance is decreased, the observed minimum distance also decreases.
%
%
%
%
This shows that Black-Box Simplex can be easily customized to a change in the safety requirement.
Doing this for the original system would require significant effort in recomputing the POMDPs and retraining the neural networks to perform a compression of the action tables.
Plots of the seven aircraft trajectories are provided in Figure~\ref{fig:acasxu7}.  Video of the 1000 ft case is available at \url{https://youtu.be/6ZXjk8k-Xqs}.

\begin{figure}[ht]
\centering
\subfloat[Original System (failure)]
{\includegraphics[width=.44\textwidth]{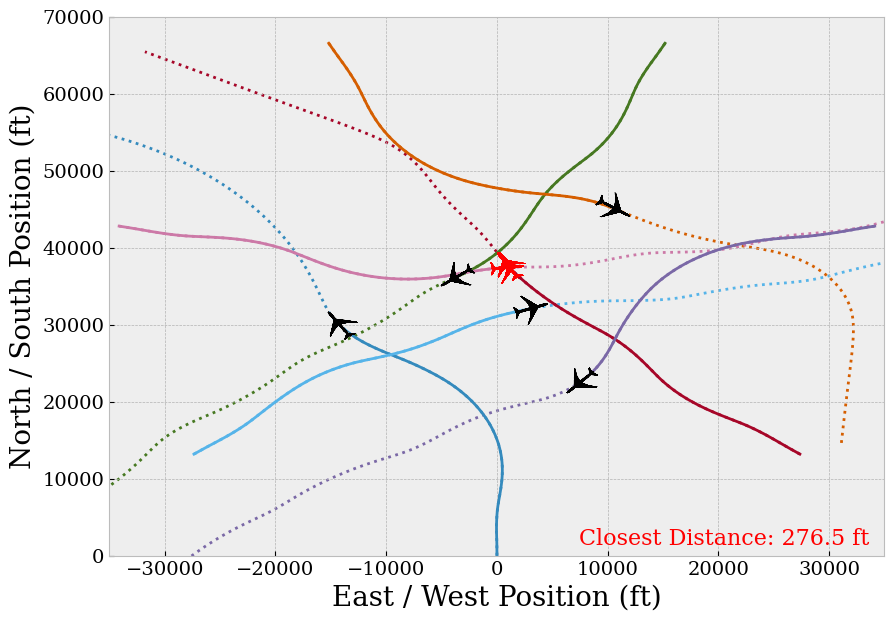}}\hfill
\subfloat[Black-Box Simplex with Safety Distance 1500 ft]
{\includegraphics[width=.44\textwidth]{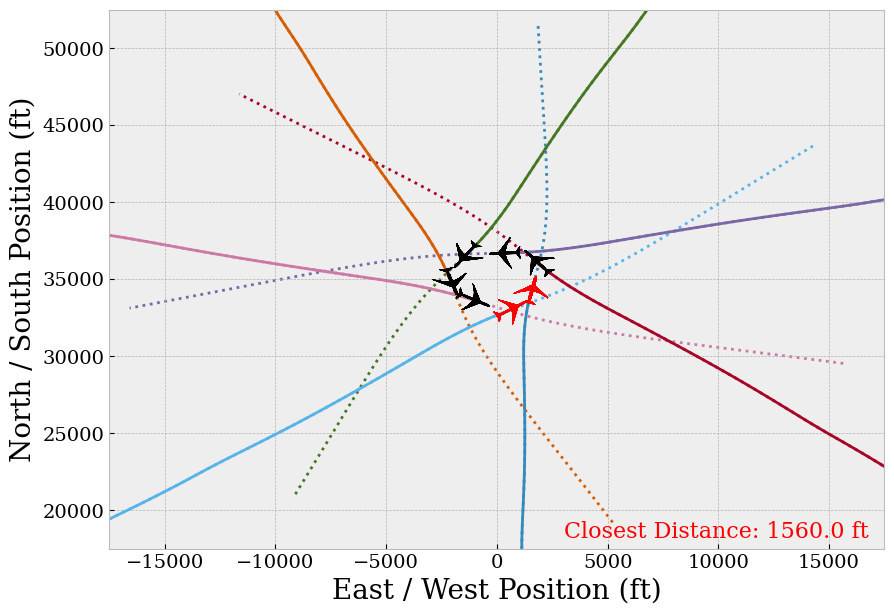}}\hfill
\subfloat[Black-Box Simplex with Safety Distance 1000 ft]
{\includegraphics[width=.44\textwidth]{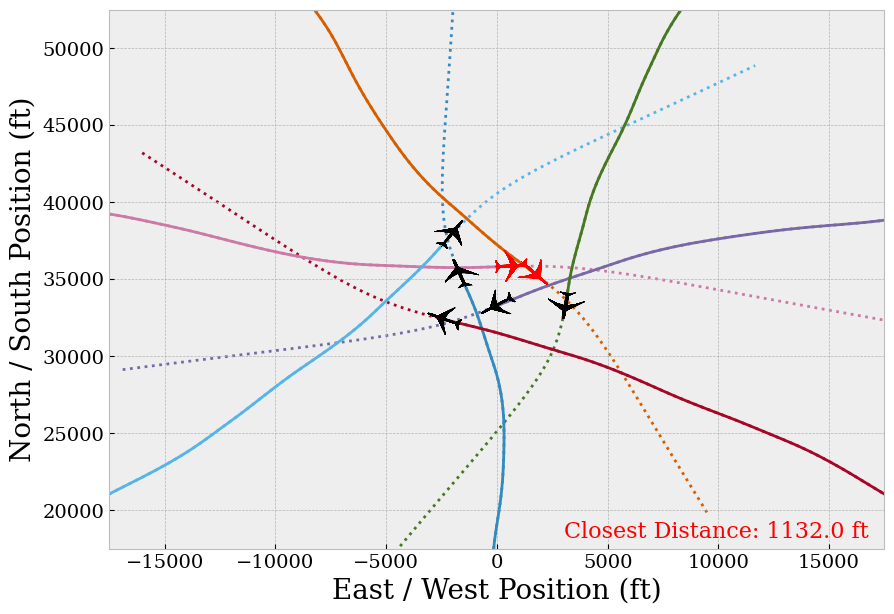}}\hfill
\subfloat[Black-Box Simplex with Safety Distance 500 ft]
{\includegraphics[width=.44\textwidth]{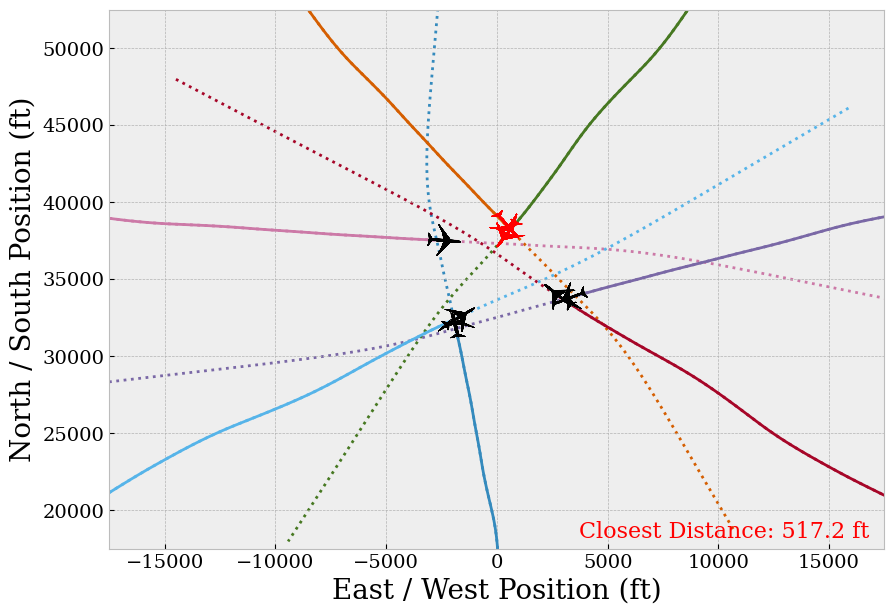}}\hfill
\caption{Black-Box Simplex is easily customizable. In the seven aircraft case, adjusting the safety distance in the decision module results in different system behaviors. In each case, the advanced controller command is overridden only enough to guarantee the corresponding safety constraint.}\label{fig:acasxu7}
\end{figure}

\end{document}